\RequirePackage[T1]{fontenc}
\documentclass[10pt,twocolumn,twoside]{IEEEtran}
\usepackage{amsmath,amssymb,amsthm,mathtools,mathrsfs,bm}
\usepackage{graphicx,booktabs,array,microtype}
\usepackage[hidelinks]{hyperref}
\usepackage[nameinlink,noabbrev]{cleveref}
\newtheorem{theorem}{Theorem}
\newtheorem{corollary}[theorem]{Corollary}
\newcommand{\R}{\mathbb R}
\newcommand{\N}{\mathbb N}
\newcommand{\one}{\mathbf 1}
\newcommand{\Prob}{\mathbb P}

\newcommand{\diag}{\operatorname{diag}}
\newcommand{\col}{\operatorname{col}}

\newcommand{\diamP}{\operatorname{diam}_P}
\newcommand{\Arg}{\operatorname{Arg}}
\allowdisplaybreaks[1]
\title{Synchronization over Resource-Causal Temporal Networks with Consumable Entanglement Links}
\author{Ruiwu Niu, Xiaoqun Wu, Guanrong Chen, and Micha\"el Antonie van Wyk%
\thanks{R. Niu is with the Department of Data Science and Digital Innovation, Hong Kong Shue Yan University, Hong Kong SAR, China (corresponding author, e-mail: rniu@hksyu.edu).}%
\thanks{X. Wu is with the College of Computer Science and Software Engineering, Shenzhen University, Shenzhen, China (e-mail: xqwu@szu.edu.cn).}%
\thanks{G. Chen is with the Department of Electrical Engineering, City University of Hong Kong, Kowloon, Hong Kong SAR, China (e-mail: eegchen@cityu.edu.hk).}%
\thanks{M. A. van Wyk is with the School of Electrical and Information Engineering, University of the Witwatersrand, Johannesburg, South Africa.}}
\begin{document}
\maketitle
\begin{abstract}
When communication consumes shared entanglement, a physical link does not guarantee that an interaction is available. We study synchronization of classical continuous-time systems whose interactions consume randomly replenished resources stored on edges. Under isotropic tangent coupling and state independent scheduling, the linearized evolution separates exactly into intrinsic node growth and an ordered product of communication matrices. This yields temporal contraction conditions and local nonlinear synchronization. Resource conservation bounds the contraction achievable by every schedule with acyclic event services. On a fixed acyclic orientation with a unique source, serving each ready edge immediately is optimal: the non-source node with the least total incoming supply determines the rate. Balanced incoming supplies attain the bound. For linear coupling on these orientations, a further condition guarantees synchronization throughout a convex invariant domain. In the specified Lorenz domain, this sufficient contraction requirement is about three percent above a necessary requirement for uniform exponential synchronization. Kuramoto simulations illustrate the effect of service direction, while Lorenz simulations compare resource policies and test local growth predictions against finite-time dynamics.
\end{abstract}

\section{Introduction}
\label{sec:introduction}
Synchronization describes convergence toward a common motion, governed partly by network structure \cite{PecoraCarroll1998,Arenas2008,WangChen2002}. Temporal interactions can differ from their static aggregate \cite{HolmeSaramaki2012,Moreau2005,Stilwell2006,Ghosh2022}; multilayer and rewiring networks also reveal the importance of interaction organization \cite{Tang2019,Wu2024,Wu2022}. Blinking and random networks admit synchronization under sufficiently fast switching \cite{Belykh2004,Porfiri2008}. Here each interaction consumes a resource that must first be generated and stored.

Quantum communication motivates this setting: transfer protocols consume stored entangled states \cite{Bennett1993,Kimble2008,Wehner2018}. Link protocols turn heralded generation into services \cite{Dahlberg2019}, allocated by routing and scheduling \cite{Pant2019,Vardoyan2021,Cicconetti2021}. Beyond communication rate, latency, and fidelity, a dynamical application asks whether completed services supply enough coupling to synchronize evolving states.

Randomized gossip relates interaction probabilities to averaging rates \cite{Boyd2006}, while consensus admits convergence-time limits for specified algorithm classes \cite{OlshevskyTsitsiklis2009}. Related control problems address sampled communication \cite{Cheng2013}, link failures \cite{Chen2021}, event triggering \cite{Dimarogonas2012,Nowzari2019}, or energy harvesting \cite{Wei2024}. Zhang and Strogatz design temporal networks under a fixed total coupling budget \cite{ZhangStrogatz2021}. In our setting, each edge stores a resource that is randomly replenished and consumed by use; a policy chooses which ready edges serve and in which direction.

Each node's classical state follows an ODE. A completed service initializes the receiver's predictor of the sender's sampled state. Physical support, inventory, and directed services thus form causally dependent layers rather than multiplex state replicas \cite{Kivela2014}. We call this structure a \emph{resource causal temporal network}. Unlike quantum oscillator synchronization \cite{Mari2013,Walter2014}, the synchronized variables are classical.

Resource supply bounds the achievable contraction, while scheduling determines how much of it is realized. Our main results are as follows.
\begin{itemize}
\item \emph{Exact separation.} Under isotropic tangent coupling and state independent scheduling, the transverse linearized evolution factorizes into intrinsic node growth and ordered communication (Theorem~\ref{thm:temporal-contraction}).
\item \emph{Resource limit.} Resource conservation gives an explicit upper bound on asymptotic contraction for every feasible schedule with acyclic services at each event. It holds almost surely, even with advance knowledge of arrivals (Theorem~\ref{thm:resource}(i)).
\item \emph{Optimal scheduling.} Immediate (greedy) service is optimal within a fixed acyclic orientation with one source. Its rate is the effective coupling strength times the minimum total incoming supply among non-source nodes; equal incoming supplies characterize attainment of the bound on the same edges (Theorem~\ref{thm:resource}).
\item \emph{Nonlinear guarantees.} Contraction conditions give local nonlinear synchronization and, under linear coupling on fixed acyclic orientations with a unique source, synchronization throughout a convex invariant domain (Corollary~\ref{cor:nonlinear-synchronization} and Theorem~\ref{thm:nonlinear}).
\end{itemize}

An ordering example and Kuramoto simulations, including the phase order parameter \cite{Niu2018}, distinguish chronology from direction. Lorenz experiments test finite-time policy performance and local growth predictions. Our ideal interface (Section~\ref{sec:entanglement-interpretation}) connects entanglement service allocation to application stability requirements.

Sections~\ref{sec:model}--\ref{sec:resource-policy} develop the theory; Sections~\ref{sec:entanglement-interpretation}--\ref{sec:numerical-results} give its interpretation and experiments. Proofs are in the appendix.

\begin{figure*}[t]
\centering
\includegraphics[width=\textwidth]{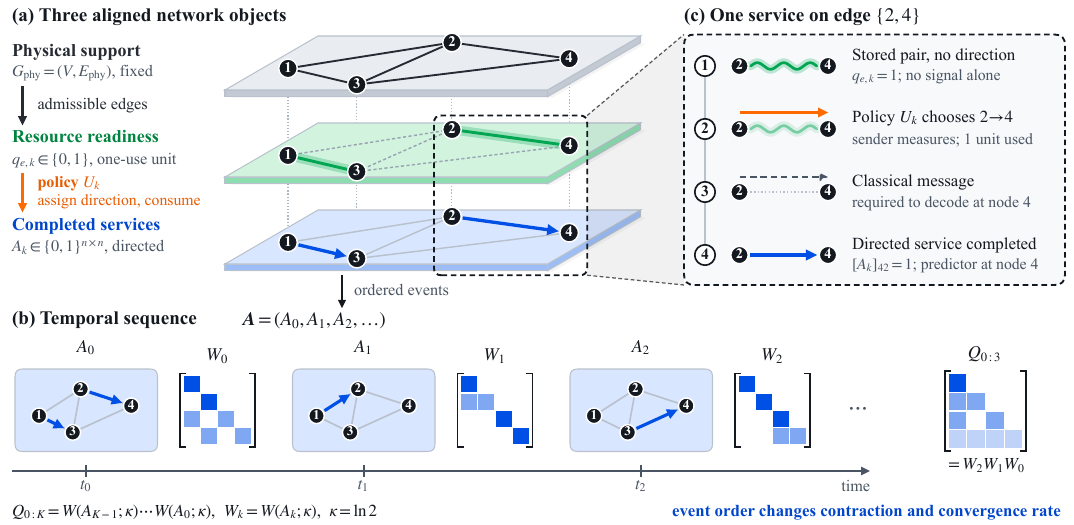}
\caption{Model architecture. (a) Fixed physical support, resource readiness, and completed directional services are aligned but distinct objects; highlighted edges hold a ready unit, and the dashed slice through edge \(\{2,4\}\) is expanded in (c). (b) The completed services form an ordered temporal sequence acting on continuous node dynamics; each event gives the row-stochastic matrix \(W_k=W(A_k;\kappa)\), drawn exactly for \(\kappa=\ln 2\) (white \(0\), blue \(1\)), and their ordered product is \(Q_{0:3}\). (c) A policy assigns a direction, completes a transfer with its classical message, consumes the stored resource, and initializes a receiver predictor.}
\label{fig:model-architecture}
\end{figure*}
\section{Model and synchronization problem}
\label{sec:model}
\subsection{Three aligned network objects}
Let \(G_{\mathrm{phy}}=(V,E_{\mathrm{phy}})\) be a fixed connected undirected graph, with \(V=\{1,\ldots,n\}\) and \(n\ge2\).
Its edges are pairs of endpoints capable of generating and storing a communication resource. Before event \(k\ge0\), the readiness variable \(q_{e,k}\in\{0,1\}\) records whether edge \(e\) holds one qualified unit. The completed service matrix \(A_k\in\{0,1\}^{n\times n}\) has zero diagonal and uses the receiver row convention: \([A_k]_{ij}=1\) means that a service \(j\to i\) completes at that event. Unsupported pairs have zero entries. The realized sequence is \(\boldsymbol A=(A_k)_{k\ge0}\).

A policy consumes ready resources and assigns directions to generate \(A_k\). Readiness does not itself imply a direction. The \emph{consumable entanglement links} are stored service resources; physical channels remain intact after use.
Writing \(U_k\) for the binary vector of directional decisions, \(g_{k+1}\) for new resource arrivals, and \(\mathcal F_k^R\) for the available resource history, the dependencies are
\begin{equation}
\begin{gathered}
G_{\mathrm{phy}}\rightsquigarrow(q_k,U_k,A_k),\qquad
(q_k,\mathcal F_k^R)\to U_k\to A_k,\\
(q_k,U_k,g_{k+1})\to q_{k+1}.
\end{gathered}
\label{eq:causal-chain}
\end{equation}
Here \(\rightsquigarrow\) means that the physical graph constrains which edges can hold resources and provide services. Section~\ref{sec:resource-policy} specifies the resource generator. Policies use resource information, not node states.

\emph{Definition (resource causal temporal network).} The fixed support, inventory process, causal resource policy, and induced service sequence, related by \eqref{eq:causal-chain}, form a resource causal temporal network. This includes the inventory feedback after consumption.

\subsection{Communication events and continuous dynamics}
Let \(\mathcal X\subseteq\R^d\) be open, \(d\ge1\), and \(f:\mathcal X\to\R^d\) be continuously differentiable. Node \(i\) has state \(x_i(t)\in\mathcal X\). Fix an event spacing \(\Delta>0\), times \(t_k=k\Delta\), and intervals \(I_k=[t_k,t_{k+1})\). We first condition on \(\boldsymbol A\). A completed service provides the ideal sample
\begin{equation}
Y_{j\to i,k}=x_j(t_k),\qquad
\widehat y_{j\to i,k}(t_k)=Y_{j\to i,k}.
\label{eq:sample}
\end{equation}
The receiving predictor then follows the isolated dynamics:
\begin{equation}
\dot{\widehat y}_{j\to i,k}=f(\widehat y_{j\to i,k}),\qquad t\in I_k.
\label{eq:predictor-flow}
\end{equation}
Node states remain continuous and satisfy
\begin{equation}
\dot x_i=f(x_i)+\sigma\sum_j[A_k]_{ij}
h_{ij}(x_i,\widehat y_{j\to i,k}),\qquad t\in I_k,
\label{eq:node-ode}
\end{equation}
where \(\sigma>0\) is the coupling gain, each \(h_{ij}:\mathcal X^2\to\R^d\) is continuously differentiable, and \(h_{ij}(z,z)=0\). Predictors evolve from sampled states, not current sender states. Section~\ref{sec:entanglement-interpretation} explains the ideal sample interface and its finite precision boundary.

\subsection{Synchronization quantities and scope}
Let \(\mathscr S\) be a nonempty class of forward complete isolated trajectories \(\dot s=f(s)\). For the exact linear reduction assume, on every supported ordered pair,
\begin{equation}
D_2h_{ij}(s(t),s(t))=bI_d,\qquad b>0,\qquad
\kappa=\sigma b\Delta.
\label{eq:isotropic-tangent}
\end{equation}
Here \(D_2\) differentiates the second argument and \(I_d\) is the identity; \(\sigma b\) is the effective coupling strength. Diagonal consistency gives \(D_1h_{ij}(s,s)=-bI_d\). This isotropic tangent condition means that coupling has the same linearized strength in every state direction, enabling exact stochastic-matrix factorization. General matrix-valued coupling needs a different analysis. The domain theorem further requires \(h_{ij}(x,y)=b(y-x)\).

Fix a symmetric positive definite matrix \(P\in\R^{d\times d}\), and write \(\|v\|_P=(v^{\mathsf T}Pv)^{1/2}\) and \(\|H\|_P=\|P^{1/2}HP^{-1/2}\|_2\), where \(\|\cdot\|_2\) is the Euclidean operator norm. The fundamental matrix \(\Psi_s(t,u)\) of \(\dot\xi=Df(s(t))\xi\), normalized at \(u\), satisfies the assumed uniform bound
\begin{equation}
\begin{gathered}
\|\Psi_s(t,u)\|_P\le M e^{\alpha_0(t-u)},\\
M\ge1,\quad\alpha_0\in\R,\quad t\ge u\ge0.
\end{gathered}
\label{eq:isolated-growth}
\end{equation}
Set
\begin{equation}
F_K=\sup_{s\in\mathscr S}\|\Psi_s(t_K,0)\|_P,\qquad
\lambda_{\mathscr S}=\limsup_{K\to\infty}\frac{\log F_K}{K\Delta}.
\label{eq:isolated-exponent}
\end{equation}
Here \(\alpha_0\) bounds growth over every interval, whereas \(\lambda_{\mathscr S}\le\alpha_0\) describes asymptotic growth.

For \(X=\col(x_1,\ldots,x_n)\), let \(\mathcal D=\{\one_n\otimes z:z\in\R^d\}\) be the synchronization subspace and define
\begin{equation}
d_P(X,\mathcal D)=\min_{z\in\R^d}
\left(\sum_i\|x_i-z\|_P^2\right)^{1/2}.
\label{eq:sync-distance}
\end{equation}
Here \(\col\) stacks vectors, \(\one_n\) is the vector of ones, and \(\otimes\) is the Kronecker product; the minimizer is the node mean. Denote the induced linearized evolution on \(\R^{nd}/\mathcal D\) by \(\mathcal J_{0:K}^{s,\boldsymbol A}\), with its norm induced by \eqref{eq:sync-distance}. Its upper transverse exponent is
\begin{equation}
\Lambda_\perp^{\mathrm{lin}}=
\limsup_{K\to\infty}\frac{1}{K\Delta}
\log\sup_{s\in\mathscr S}\|\mathcal J_{0:K}^{s,\boldsymbol A}\|_{P,\perp}.
\label{eq:linear-exponent}
\end{equation}
A negative value means exponential decay of linearized disagreement. Nonlinear synchronization requires
\begin{equation}
d_P(X(t),\mathcal D)\le C_{\boldsymbol A,\lambda}e^{-\lambda t}
d_P(X(0),\mathcal D),\qquad t\ge0,
\label{eq:nonlinear-target}
\end{equation}
with \(\lambda>0\) and a finite constant independent of initial states in the stated neighborhood or invariant domain, but possibly dependent on the service sequence.

\section{From temporal coupling to synchronization}
\label{sec:temporal-theory}

We first ask whether a given service sequence can synchronize the nodes, without specifying how resources generate it. Exact event matrices separate the contraction supplied by communication from the growth of isolated node differences.

\subsection{Exact reduction and temporal contraction}

For a service matrix \(A\), let \(\ell_i=\sum_j A_{ij}\in\{0,\ldots,n-1\}\) be the number of received samples, written \(\ell_{i,k}\) when \(A=A_k\), and define
\begin{equation}
 W(A;\kappa)_{ij}=\begin{cases}
 e^{-\kappa\ell_i},&i=j,\ \ell_i>0,\\
 (1-e^{-\kappa\ell_i})A_{ij}/\ell_i,&i\ne j,\ \ell_i>0,\\
 1,&i=j,\ \ell_i=0,\\
 0,&\text{otherwise}.
 \label{eq:service-matrix}
\end{cases}
\end{equation}
Its entries are nonnegative and its rows sum to one. Write
\begin{equation}
 Q_{0:K}=W(A_{K-1};\kappa)\cdots W(A_0;\kappa),
 \qquad Q_{0:0}=I_n.
 \label{eq:ordered-factorization}
\end{equation}
The earliest service is on the right. Choose an orthonormal basis \(U\in\R^{n\times(n-1)}\) of \(\one_n^\perp\), and set
\begin{equation}
 \|Q\|_\perp=\|U^{\mathsf T}QU\|_2,\qquad
 \Lambda_Q=\limsup_{K\to\infty}\frac{\log\|Q_{0:K}\|_\perp}{K\Delta}.
 \label{eq:network-exponent}
\end{equation}
The quotient norm is independent of the basis \(U\). We call \(c:=-\Lambda_Q\) the \emph{network contraction rate}: larger \(c\) means faster asymptotic suppression of disagreement by communication. This rate must be compared with the intrinsic growth exponent \(\lambda_{\mathscr S}\) in \eqref{eq:isolated-exponent}. For a row-stochastic matrix \(Q=(q_{ij})\), its Dobrushin coefficient is
\begin{equation}
 \delta(Q)=\frac12\max_{r,\ell}\sum_j|q_{rj}-q_{\ell j}|\in[0,1].
 \label{eq:dobrushin}
\end{equation}
It measures the largest remaining difference between receiving rows and satisfies \(\delta(QR)\le\delta(Q)\delta(R)\) \cite{Seneta2006}. We use \(\log0=-\infty\).

Contraction depends on the ordered product, not just a static aggregate. The finite-window condition below makes this dependence usable for policy design through window duration, contraction strength, and occurrence frequency.

\begin{theorem}[Exact temporal reduction and contraction]
\label{thm:temporal-contraction}
\label{lem:exact-service-matrix}
\label{lem:separation}
Under the isotropic tangent condition \eqref{eq:isotropic-tangent}, fix a service sequence independently of the node states. Then
\begin{equation}
 \sup_{s\in\mathscr S}\|\mathcal J_{0:K}^{s,\boldsymbol A}\|_{P,\perp}
 =\|Q_{0:K}\|_\perp F_K,
 \label{eq:exact-norm}
\end{equation}
and consequently
\begin{equation}
 \Lambda_\perp^{\rm lin}\le\lambda_{\mathscr S}+\Lambda_Q
 \le\alpha_0+\limsup_K\frac{\log\delta(Q_{0:K})}{K\Delta}.
 \label{eq:general-contraction}
\end{equation}
The first inequality is an equality if the normalized logarithm of either factor in \eqref{eq:exact-norm} has a finite limit.

More explicitly, fix a block length \(L\in\N\) and define
\begin{equation}
 B_r=W(A_{(r+1)L-1};\kappa)\cdots W(A_{rL};\kappa).
 \label{eq:block-product}
\end{equation}
If \(\gamma\in[0,1)\), \(\nu\in(0,1]\), and
\begin{equation}
 \liminf_{N\to\infty}\frac1N\sum_{r=0}^{N-1}
 \mathbf1\{\delta(B_r)\le\gamma\}\ge\nu,
 \label{eq:block-frequency}
\end{equation}
then
\begin{equation}
 \Lambda_\perp^{\rm lin}\le\alpha_0+\frac{\nu\log\gamma}{L\Delta}.
 \label{eq:frequency-bound}
\end{equation}
Here \(\mathbf1\{\cdot\}\) denotes the indicator of an event.
\end{theorem}

The proof in Appendix~\ref{app:exact-proof} integrates each receiver in isolated variational coordinates, then bounds the ordered product by its diameter contraction coefficient.

For a fixed trajectory whose isolated logarithmic growth rate has a finite limit, the transverse linear exponent is that limit plus \(\Lambda_Q\).

\subsection{Two nonlinear consequences}
\label{sec:nonlinear-consequence}

The linear reduction measures infinitesimal disagreement. To relate it to the nonlinear nodes, we first obtain a local synchronization guarantee for coupling satisfying \eqref{eq:isotropic-tangent}. Let
\begin{equation}
 \beta_L=\sup_{r\ge0,s\in\mathscr S}
 \|\Psi_s(t_{(r+1)L},t_{rL})\|_P
 \le M e^{\alpha_0L\Delta}.
 \label{eq:block-growth}
\end{equation}
Appendix~\ref{app:nonlinear-proof} requires uniform flow and differentiability bounds near synchronization and closure under restarting comparison trajectories. A negative linear exponent alone does not imply these conditions.

\begin{corollary}[Local nonlinear synchronization]
\label{cor:nonlinear-synchronization}
Under \eqref{eq:block-frequency} and the local regularity conditions in Appendix~\ref{app:nonlinear-proof},
\begin{equation}
 \log\beta_L+\nu\log\gamma<0
 \label{eq:nonlinear-condition}
\end{equation}
implies local exponential synchronization in the sense of \eqref{eq:nonlinear-target}. The bound \eqref{eq:block-growth} can be used when \(\beta_L\) is not known directly. The neighborhood and convergence constant may depend on the service realization.
\end{corollary}

To cover finite initial differences throughout a specified domain, we next compare actual node differences with an isolated source trajectory. Fix a directed acyclic graph (DAG) \(D\) on \(V\), using supported edges and having a unique source \(o\), meaning the only node with no incoming edge. Restrict every service to this orientation. Unlike an outward tree, this permits several incoming links at a node. Suppose \(h_{ij}(x,y)=b(y-x)\), \(\mathcal X\) is convex, and a convex \(C^1\) function \(\mathcal V:\mathcal X\to\R\) has a nonempty compact sublevel set
\begin{equation}
 \Omega=\{z\in\mathcal X:\mathcal V(z)\le v_*\},\qquad v_*\in\R.
 \label{eq:domain}
\end{equation}
Assume that, for some \(a>0\),
\begin{equation}
 \nabla\mathcal V(z)^{\mathsf T}f(z)
 \le-a(\mathcal V(z)-v_*),\qquad z\in\mathcal X,
 \label{eq:invariance-assumption}
\end{equation}
and that some \(\alpha_\Omega\in\R\) satisfies
\begin{equation}
 \langle x-y,f(x)-f(y)\rangle_P
 \le\alpha_\Omega\|x-y\|_P^2,\qquad x,y\in\Omega.
 \label{eq:one-sided}
\end{equation}
Here \(\langle u,v\rangle_P=u^{\mathsf T}Pv\). A sufficient test is
\begin{equation}
 \mu_P(Df(z)):=\lambda_{\max}\!\left(
 \operatorname{sym}(P^{1/2}Df(z)P^{-1/2})
 \right)\le\alpha_\Omega
 \label{eq:lognorm}
\end{equation}
on \(\Omega\), where \(\operatorname{sym}(H)=(H+H^{\mathsf T})/2\) and \(\lambda_{\max}\) denotes the largest eigenvalue. Integrating the Jacobian along a segment in the convex domain proves \eqref{eq:one-sided}.

\begin{theorem}[Synchronization on a convex invariant domain]
\label{thm:nonlinear}
Under \eqref{eq:domain}--\eqref{eq:one-sided}, for a state independent service sequence on the fixed DAG, all solutions starting in \(\Omega^n\) are forward complete and all node states and active predictors remain in \(\Omega\). If
\begin{equation}
 \alpha_\Omega+\Lambda_Q<0,
 \label{eq:nonlinear-rate-condition}
\end{equation}
then \eqref{eq:nonlinear-target} holds for every \(X(0)\in\Omega^n\) and every rate \(0<\lambda<-\alpha_\Omega-\Lambda_Q\), with a finite multiplicative constant independent of the initial state. The nodes synchronize to the source's isolated trajectory.
\end{theorem}

Appendix~\ref{app:domain-proof} proves the domain guarantee; Section~\ref{sec:resource-policy} supplies computable policy rates. The equilibrium test below assesses conservatism by giving a necessary condition for the same uniform guarantee.

\emph{Remark (a necessary test for a uniform guarantee).}
\label{rem:equilibrium-necessity}
Let \(z_*\in\operatorname{int}\Omega\) satisfy \(f(z_*)=0\), and let \(a_*\in\R\) be the largest real part of an eigenvalue of \(Df(z_*)\). For a fixed service realization independent of node states, with isotropic tangent at \(z_*\), estimate \eqref{eq:nonlinear-target} with the same finite constant and rate \(\lambda>0\) for all initial states in \(\Omega^n\) necessarily implies
\begin{equation}
 a_*+\Lambda_Q\le-\lambda<0.
 \label{eq:equilibrium-necessity}
\end{equation}
The proof differentiates the uniform estimate at the equilibrium (Appendix~\ref{app:domain-proof}).

\subsection{Why chronological order matters}

The following example isolates event order while keeping the other factors fixed. Consider three scalar nodes on the undirected path \(1\! -\!2\! -\!3\), with \(f=0\), \(h(x,y)=y-x\), and \(\kappa=\log2\). Compare the chronological blocks
\begin{equation}
 \begin{aligned}
 \mathcal O_{\rm prop}&=(1\to2,\ 2\to3,\ 3\to2,\ 2\to1),\\
 \mathcal O_{\rm alt}&=(2\to3,\ 1\to2,\ 2\to1,\ 3\to2).
 \end{aligned}
 \label{eq:order-products}
\end{equation}
Let \(Q_{\rm prop}\) and \(Q_{\rm alt}\) be their ordered products. They have the same service multiset, duration, strength, and physical support, but \(\delta(Q_{\rm prop})=7/16\) and \(\delta(Q_{\rm alt})=3/4\). Their transverse spectral radii, meaning the largest eigenvalue moduli excluding the synchronized eigenvalue \(1\), are
\begin{equation}
 \rho_\perp(Q_{\rm prop})=\frac{9+\sqrt{17}}{32},\qquad
 \rho_\perp(Q_{\rm alt})=\frac{3+\sqrt5}{8}.
 \label{eq:order-spectral-values}
\end{equation}
Repeating these blocks gives distinct contraction rates despite identical static aggregates. Direct iteration of the two block maps from the same initial state confirms the predicted difference in disagreement decay (not shown). The fixed-DAG frequency formula below exploits a special triangular structure; it does not remove this general order effect.

\section{From consumable resources to temporal coupling}
\label{sec:resource-policy}
We now ask how resources can generate the contracting service sequences identified above. A frame policy gives a construction; a resource bound then limits the contraction achievable by feasible policies with acyclic event services.

\subsection{Inventory and feasible decisions}
For \(e=\{i,j\}\), let \(u_{j\to i,k},u_{i\to j,k}\in\{0,1\}\) be the directional decisions and impose
\begin{equation}
s_{e,k}:=u_{j\to i,k}+u_{i\to j,k}\le q_{e,k},\qquad
[A_k]_{ij}=u_{j\to i,k}.
\label{eq:feasible-action}
\end{equation}
Each edge serves at most one direction per event; distinct edges may operate simultaneously even when sharing a node. A causal policy uses the resource history \(\mathcal F_k^R=\operatorname{sig}(q_0,U_0,q_1,\ldots,U_{k-1},q_k)\), optionally enlarged by independent policy randomization available by event \(k\), but not node states or future arrivals. After consumption and interval evolution, arrivals \(g_{e,k+1}\in\{0,1\}\) update inventory:
\begin{equation}
q_{e,k+1}=\min\{1,q_{e,k}-s_{e,k}+g_{e,k+1}\}.
\label{eq:inventory-update}
\end{equation}
Arrivals are independent across edges and positive event indices, with \(\Prob(g_{e,k}=1)=p_e\in[0,1]\); \(q_0\) is arbitrary. A new unit is first available at the next decision. Storage has unit capacity and no loss. These assumptions define the policy class in the rate comparison below.

\subsection{Accumulation and dissemination as a construction}
Choose a spanning tree \(T\) and a root \(o\). Its outward orientation directs each parent to its children. Let \(h_T\ge1\) be its maximum depth, and let \(A_T^{(d)}\) contain the edges whose children are at depth \(d\). Their ordered product is
\begin{equation}
B_T=W(A_T^{(h_T)};\kappa)\cdots W(A_T^{(1)};\kappa).
\label{eq:tree-product}
\end{equation}
Writing \(\eta=1-e^{-\kappa}\) and \(d_T(i)\) for node \(i\)'s depth, the root weight in row \(i\) is \(\eta^{d_T(i)}\ge\eta^{h_T}\). Hence
\begin{equation}
\delta(B_T)\le\gamma_T:=1-(1-e^{-\kappa})^{h_T}<1.
\label{eq:tree-contraction}
\end{equation}
For an integer \(\tau\ge1\), divide time into frames of \(L=\tau+h_T\) events. The \emph{frame policy} \(\pi_{T,\tau}\) waits for the first \(\tau\) events, checks whether every tree edge is ready, and, if so, executes the \(h_T\) levels in order. Otherwise it idles for the rest of the frame. Edges outside \(T\) are unused.

Fresh arrivals fill every edge during the accumulation phase with probability
\begin{equation}
\rho_{T,\tau}=\prod_{e\in T}[1-(1-p_e)^\tau].
\label{eq:ready-probability}
\end{equation}
These fresh-readiness events depend on disjoint arrival sets in different frames. Each guarantees success regardless of the previous inventory, so the strong law gives a success frequency at least \(\rho_{T,\tau}\), simultaneously for all initial inventories. If \(\rho_{T,\tau}>0\), Theorem~\ref{thm:temporal-contraction} applies with \(\nu=\rho_{T,\tau}\) and \(\gamma=\gamma_T\), giving \(c\ge-\rho_{T,\tau}\log\gamma_T/(L\Delta)\). With the same substitution, \eqref{eq:nonlinear-condition} supplies a sufficient test under its local regularity assumptions. The frame policy thus realizes the contracting windows required by Corollary~\ref{cor:nonlinear-synchronization} and supplies the schedule used in Fig.~\ref{fig:kuramoto-direction}.

\subsection{Resource limit and an attaining policy}
To assess the costs of waiting and sequential dissemination, we next ask what asymptotic contraction the resource supply permits. Let \(S_e(K)=\sum_{k=0}^{K-1}s_{e,k}\) count services. A tree can leave useful physical edges unused, so we consider a fixed directed acyclic graph \(D\) on \(V\), orienting an allowed edge set \(F\subseteq E_{\mathrm{phy}}\), with unique source \(o\). Write \(\operatorname{in}_D(i)\) for node \(i\)'s incoming edges, \(d_i^-=|\operatorname{in}_D(i)|\), and \(S_i(K)=\sum_{e\in\operatorname{in}_D(i)}S_e(K)\). Greedy \(\mathcal G_D\) serves every ready edge in its prescribed direction. The theorem first bounds all acyclic event services, then identifies the optimum within this fixed orientation.

\begin{theorem}[Resource limit and optimal scheduling on a fixed DAG]
\label{thm:resource}
There is one arrival event of probability one on which the following statements hold for all initial inventories.

(i) For any fixed allowed set \(F\subseteq E_{\mathrm{phy}}\), every feasible sequence on \(F\) whose service graph is acyclic at each event satisfies
\begin{equation}
\liminf_{K\to\infty}\frac{\log\|Q_{0:K}\|_\perp}{K\Delta}
\ge-\frac{\sigma b}{n-1}\sum_{e\in F}p_e.
\label{eq:resource-bound}
\end{equation}
The same lower bound holds for \(\Lambda_Q\), including schedules that know future arrivals.

(ii) For services using the fixed orientation \(D\), if \(R_i=\lim_{K\to\infty}S_i(K)/K\in[0,d_i^-]\) exists for every \(i\ne o\), then
\begin{equation}
\lim_{K\to\infty}\frac{\log\|Q_{0:K}\|_\perp}{K\Delta}
=-\sigma b\min_{i\ne o}R_i.
\label{eq:dag-rate}
\end{equation}
This identity is deterministic, including zero frequencies.

(iii) Put \(w_i=\sum_{e\in\operatorname{in}_D(i)}p_e\). Greedy service has edge frequencies \(p_e\) and gives
\begin{equation}
c_D:=-\Lambda_Q(\mathcal G_D)
=\sigma b\min_{i\ne o}w_i.
\label{eq:greedy-dag}
\end{equation}
It maximizes \(c\) over all feasible schedules using \(D\), even when competing frequencies have no limits. It attains the total supply bound \eqref{eq:resource-bound} for the same \(F\) if and only if all \(w_i\), \(i\ne o\), are equal.
\end{theorem}

The proof in Appendix~\ref{app:resource-proof} uses conservation for the supply limit and triangular products for the rate on a fixed orientation. Part (i) requires acyclic services, not acyclic physical support: their orientation may change between events. Parts (ii)--(iii) use one fixed orientation throughout.

On an outward tree, each nonroot node has one incoming edge. Thus, whenever the edge frequencies \(r_e=\lim_K S_e(K)/K\in[0,1]\) exist, part (ii) reduces to \(c=\sigma b\min_{e\in T}r_e\); greedy gives \(c=\sigma b\min_{e\in T}p_e\). If all tree edges have the same arrival probability \(p\) and only those edges are allowed, every directional service is acyclic. Greedy then attains the bound \(c\le\sigma bp\) from part (i), making it optimal among \emph{all} feasible directional schedules on that tree. This is the comparison in Fig.~\ref{fig:resource-comparison}.

On a fixed DAG, the bottleneck is instead the minimum total incoming supply. Equal incoming supplies need not require equal edge probabilities: probabilities \(0.6\), \(0.2\), and \(0.4\) on \(1\to2\), \(1\to3\), and \(2\to3\) give two incoming sums of \(0.6\). Figure~\ref{fig:general-graph} uses \eqref{eq:greedy-dag} to compare trees with orientations retaining additional edges.

Combining the finite resource bound with the exact separation gives, for the acyclic services in part (i),
\begin{equation}
\Lambda_\perp^{\mathrm{lin}}\ge
\lambda_{\mathscr S}-\frac{\sigma b}{n-1}\sum_{e\in F}p_e.
\label{eq:necessary-resource}
\end{equation}
A negative transverse linear exponent therefore requires \(\sigma b\sum_{e\in F}p_e>(n-1)\lambda_{\mathscr S}\). If \(\Omega\) contains the interior equilibrium from \eqref{eq:equilibrium-necessity}, the uniform nonlinear guarantee in \eqref{eq:nonlinear-target} similarly requires
\begin{equation}
\frac{\sigma b}{n-1}\sum_{e\in F}p_e>a_*.
\label{eq:necessary-uniform-resource}
\end{equation}
This necessary, not sufficient, condition covers every state independent schedule in part (i), including every directional schedule on a physical tree.

The resource layer supplies \(c\), while the synchronization criteria specify how large it must be. For a fixed DAG with limiting incoming frequencies, the exact linear exponent is \(\lambda_{\mathscr S}-c\). Under the assumptions of Theorem~\ref{thm:nonlinear}, \(c>\alpha_\Omega\) is sufficient for the uniform domain guarantee; \(c>a_*\) is necessary when the stated equilibrium lies in its interior.

For a tree, depth and event order affect transients, not the rate at fixed frequencies; waiting can change the frequencies themselves. With weights \(p_e\), a maximum total weight spanning tree also maximizes its minimum edge weight, hence the greedy rate among outward spanning trees. Root choice affects transients but not this rate.

\section{Entanglement assisted service interpretation}
\label{sec:entanglement-interpretation}
The model begins with completed transfers and does not prescribe a quantum carrier. A hybrid node combines quantum memory, a quantum interface, and a classical controller that accepts qualified transfers \cite{Dahlberg2019,Cicconetti2021}.

A service encodes a finite precision sample, consumes entanglement, and decodes the sample after local operations and exchange of the classical outcome \cite{Bennett1993}. Ideal teleportation of a \(D\)-dimensional carrier, \(D\ge2\), consumes a maximally entangled state of that dimension and uses \(D^2\) possible outcomes. Thus one resource unit may represent the resources needed for a full transfer.

Finite codebooks quantize real-valued samples; continuous variable teleportation adds noise at finite squeezing \cite{BraunsteinKimble1998,Pirandola2015}. Equation~\eqref{eq:sample} is the ideal limit that isolates temporal organization. In practice, sampling precedes completion; imprecise samples, erasures, and latency require an extension of this interface.

Physical edges remain undirected; policies assign directions and consume stored resources. Classical messages contribute to completion delay, consistently with no signaling \cite{Simon2001}.

Binary inventory assumes a qualified unit stored without degradation. Decoherence, purification, swapping, larger memories, and contention require a richer generator subject to the temporal synchronization requirement.

\section{Numerical illustrations}
\label{sec:numerical-results}
Each comparison keeps the physical support fixed and uses resource histories as stochastic replicates. We use fourth-order Runge--Kutta integration with steps aligned to service events. Detailed checks are reported in the accompanying Numerical Verification supplement.

\subsection{A nonlinear resource tree}
\label{sec:kuramoto-numerics}
Take seven identical-frequency Kuramoto oscillators on a fixed undirected binary tree of depth two, rooted at node 1 \cite{Acebron2005}. On a common phase lift in the rotating frame,
\begin{equation}
\dot x_i=\sigma\sum_j[A_k]_{ij}
\sin(\widehat x_{j\to i,k}-x_i).
\label{eq:kuramoto-model}
\end{equation}
Here \(f=0\), \(h(x,y)=\sin(y-x)\), \(b=1\), and predictors are constant. Taking \(\mathscr S\) as all constant synchronous states gives \(P=M=\beta_L=1\) and \(\alpha_0=0\). Bounded sine derivatives, phase shift invariance, and invariance of a lifted interval of width below \(\pi\) verify the regularity and restart requirements of Corollary~\ref{cor:nonlinear-synchronization}.
\begin{figure*}[!t]
\centering
\includegraphics[width=\textwidth]{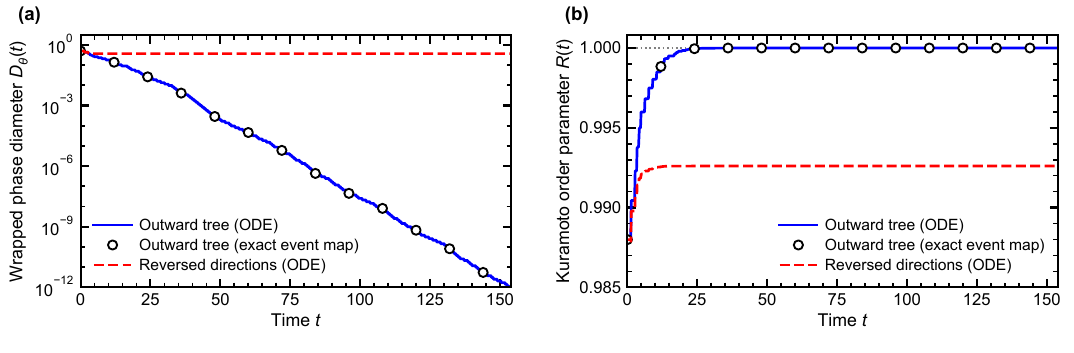}
\caption{Kuramoto synchronization on a resource tree. (a) Outward services remove the wrapped phase diameter; open circles show the exact event map. Reversing directions at the same service times leaves unequal leaf phases unchanged. (b) The phase order parameter approaches one in the outward run. Both runs share physical support, arrivals, inventory, and edge consumption times.}
\label{fig:kuramoto-direction}
\end{figure*}

The frame policy serves the two levels in order. For common \(p>0\),
\[
\rho_{T,\tau}=[1-(1-p)^\tau]^6,\qquad
\gamma_T=1-(1-e^{-\kappa})^2.
\]
The local condition \(\rho_{T,\tau}\log\gamma_T<0\) holds for all \(p>0\). With \(\sigma=2\), \(\Delta=0.1\), \(p=0.3\), \(\tau=5\), and empty initial inventories, 150 of 220 frames complete both levels. We reverse directions while retaining these service times.

With \(\Arg\in(-\pi,\pi]\) the principal argument, define \(D_\theta=\max_{i,j}|\Arg(e^{\mathrm i(x_i-x_j)})|\) and the order parameter
\begin{equation}
R(t)=\left|\frac1n\sum_j e^{\mathrm i x_j(t)}\right|\in[0,1],
\qquad \mathrm i^2=-1.
\label{eq:kuramoto-order-parameter}
\end{equation}
From \(x(0)=(0,-0.24,0.12,-0.11,0.04,0.15,0.26)\), outward services give \(D_\theta<10^{-12}\) and \(R=1\) at numerical precision. Reversed services leave unequal leaves unforced, giving \(D_\theta=0.37\), \(R\approx0.993\). This isolates direction, not order, and illustrates the local guarantee.

We check the integration against the exact receiver map
\[
u^+=v+2\arctan\!\left[e^{-\kappa}\tan((u-v)/2)\right].
\]
Here \(u\) is the receiver state and \(v\) its fixed predictor.

\subsection{Lorenz synchronization and two kinds of prediction}
\label{sec:lorenz-numerics}
To distinguish local decay from a guarantee throughout a domain, we use five Lorenz nodes \cite{Lorenz1963} on an undirected star, with root 1, outward services, linear coupling \(h(x,y)=y-x\), \(\sigma=32\), and \(\Delta=0.05\). The isolated node vector field is
\[
f_{\mathrm L}(x,y,z)=(10(y-x),\,28x-y-xz,\,xy-8z/3).
\]
\begin{figure*}[!t]
\centering
\includegraphics[width=\textwidth]{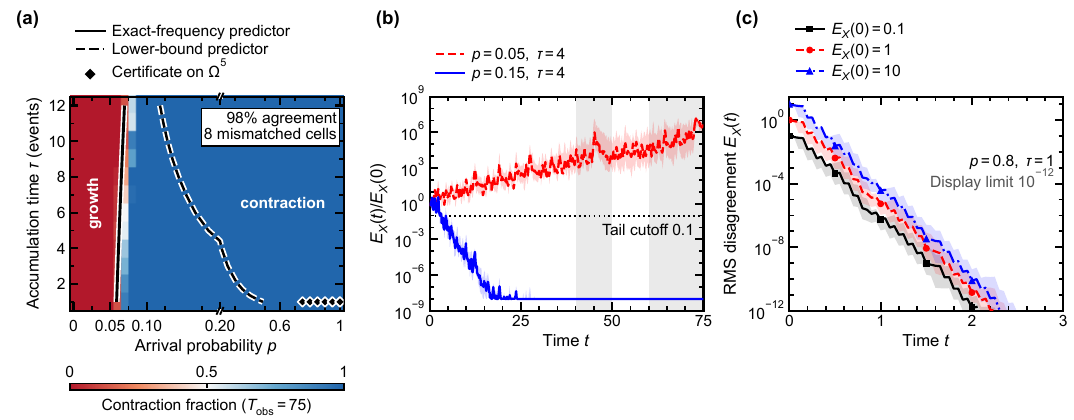}
\caption{Lorenz dynamics on one physical star. (a) Contraction fraction over 32 histories at horizon 75. Solid and dashed contours use the empirical isolated exponent with exact and lower-bound frequencies, respectively; diamonds mark the nonlinear certificate on \(\Omega^5\). Breaks mark the axis scale change at \(p=0.2\). (b) Median and interquartile relative errors. The dotted line marks the tail cutoff; shading marks final observation windows. (c) Finite initial differences at \(p=0.8,\tau=1\). Curves and bands show the median and interquartile range over 32 histories, after taking the median over 16 configurations per history. Display ends at time 3 and omits median errors below \(10^{-12}\).}
\label{fig:lorenz-numerics}
\end{figure*}
After each accumulation interval, the frame policy serves all four edges if they are ready. For common arrival probability \(p\), the fresh-readiness bound is \(\rho=[1-(1-p)^\tau]^4\). The exact stationary fraction of successful frames is
\begin{equation}
\nu_\star(p,\tau)=\left[
\sum_{r=1}^4\frac{(-1)^{r+1}\binom4r}
{1-(1-p)^{r(\tau+1)}}\right]^{-1},\qquad p>0,
\label{eq:exact-service-frequency}
\end{equation}
with \(\nu_\star(0,\tau)=0\). After service, each edge refills within a frame with probability \(1-(1-p)^{\tau+1}\). The maximum of four independent geometric refill times has the bracketed mean, whose reciprocal gives the renewal frequency. Thus \(r_e=\nu_\star/(\tau+1)\).

Exact separation gives the trajectory-specific linear threshold \(\lambda_1-\sigma\nu_\star/(\tau+1)=0\) when the isolated exponent \(\lambda_1\) exists. The solid contour uses the empirical estimate \(\widehat\lambda_1\approx0.91\); the dashed contour also replaces \(\nu_\star\) by \(\rho\). Neither contour is a domain guarantee.

We now apply Theorem~\ref{thm:nonlinear} to obtain a guarantee throughout a specified Lorenz domain. Put \(w=z-38\), \(\mathcal V=x^2+y^2+w^2\), \(R_{\mathrm L}^2=23104/15\), and \(\Omega=\{\mathcal V\le R_{\mathrm L}^2\}\). Completing the square gives
\begin{equation}
\dot{\mathcal V}=-2\mathcal V+\frac{46208}{15}-18x^2
-\frac{10}{3}\left(w+\frac{152}{5}\right)^2
\le-2(\mathcal V-R_{\mathrm L}^2).
\label{eq:lorenz-invariant-region}
\end{equation}
For \(P=\diag(3,1,1)\), the symmetric transformed Jacobian is
\[
\begin{pmatrix}
-10&(20-w)/(2\sqrt3)&y/(2\sqrt3)\\
(20-w)/(2\sqrt3)&-1&0\\
y/(2\sqrt3)&0&-8/3
\end{pmatrix}.
\]
On \(\Omega\), the norm of its first-row off diagonal vector is at most \((20+R_{\mathrm L})/(2\sqrt3)\). Bounding the lower block by \(-I_2\) yields
\begin{equation}
\alpha_\Omega=
\frac{-11+\sqrt{81+(20+152/\sqrt{15})^2/3}}2
\approx12.2.
\label{eq:lorenz-uniform-alpha}
\end{equation}
Theorem~\ref{thm:nonlinear} therefore certifies all initial states in \(\Omega^5\) when
\begin{equation}
\frac{32\nu_\star(p,\tau)}{\tau+1}>\alpha_\Omega.
\label{eq:lorenz-nonlinear-certificate}
\end{equation}
The origin lies inside this domain because \(\mathcal V(0)=1444<R_{\mathrm L}^2\). Its largest unstable eigenvalue is
\[
 a_*=\frac{-11+\sqrt{1201}}2\approx11.8.
\]
Thus \eqref{eq:equilibrium-necessity} requires \(c=-\Lambda_Q>a_*\) for a uniform exponential guarantee on \(\Omega^5\), for each fixed service realization. The sufficient requirement \(c>\alpha_\Omega\) is about \(3\%\) higher. The empirical contour instead concerns a typical chaotic trajectory.

We test \(p=0,0.01,\ldots,0.15,0.20,0.25,\ldots,1\) and \(\tau=1,\ldots,12\) with 32 independent histories paired across cells. All nodes start in \(\Omega\), with a fixed reference and centered perturbation. Define
\[
E_X(t)=\left[n^{-1}\sum_i\|x_i(t)-\overline x(t)\|_2^2\right]^{1/2}.
\]
Here \(\overline x=n^{-1}\sum_i x_i\) and \(E_X(0)=10^{-6}\). At horizon \(T_{\rm obs}\), a path passes if terminal logarithmic growth is negative (error floor \(10^{-14}\)) and the ninetieth percentile of \(E_X(t)/E_X(0)\) over the last fifth is at most \(0.1\). At least 16 of 32 paths must pass for a cell to be labeled contracting.

At horizon 75, the contour agrees in 388/396 cells (98\%), compared with 379 (96\%) at horizon 50. All discrepancies are predicted contractions failing the finite-time rule. Agreement falls to 38/46 (83\%) within \(|\widehat\lambda_1-c|\le0.25\). Paired cells are not independent replicates. All histories pass at both horizons in six certified cells: \(\tau=1\), \(p=0.75,0.80,\ldots,1\).

Figure~\ref{fig:lorenz-numerics}(c) uses \(p=0.8,\tau=1\), with rate \(13.8>\alpha_\Omega\). Sixteen uniform anchors in the radius-8 ball about \((0,0,38)\) receive centered Gaussian offsets scaled to \(E_X(0)=0.1,1,10\). Since \(8+10\sqrt5<R_{\mathrm L}\), all initial states lie in \(\Omega^5\). With 32 paired histories, all 1536 combinations remain in the domain and reach numerical precision by time 10.

\subsection{From resource rates to Lorenz response}
Having identified the local and domain-wide Lorenz contraction requirements, Fig.~\ref{fig:resource-comparison} compares how greedy service and frame policies \(\pi_{T,1}\), \(\pi_{T,5}\) meet them under the same supply on the star. Their network contraction rates are \(32p\) and \(32\nu_\star/(\tau+1)\), respectively. Theorem~\ref{thm:resource} gives \(c\le32p\) for every directional schedule, so the black line is both the supply limit and the optimum.

Panel (a) estimates \(32\min_e S_e(K)/K\) from 32 histories paired across policies and probabilities. Greedy's finite-time minimum is downward biased; its error bar does not bound the asymptotic line.
\begin{figure*}[!t]
\centering
\includegraphics[width=\textwidth]{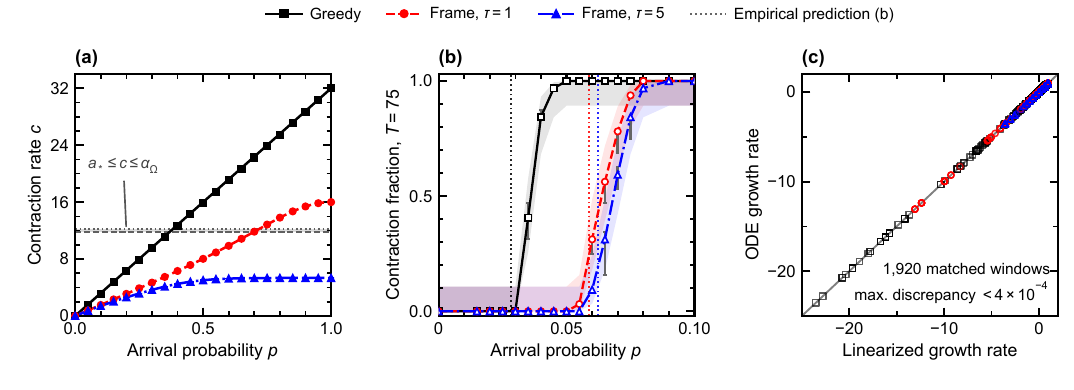}
\caption{Resource rates and Lorenz response. (a) Lines give exact rates; markers show means and one standard error over 32 histories at 20000 events. The strip spans necessary \(a_*\) to sufficient \(\alpha_\Omega\) requirements for a domain guarantee. (b) Horizon-75 ODE contraction fractions. Colored bands show pointwise Wilson 95\% intervals over 32 histories; gray segments show ranges across three integration resolutions at 12 tested cells. Dotted lines use the empirical isolated exponent. (c) Predicted versus measured growth over 1920 matched windows and initial vectors, using the maximum root-relative norm rather than Fig.~\ref{fig:lorenz-numerics}'s RMS error.}
\label{fig:resource-comparison}
\end{figure*}

\begin{figure*}[!t]
\centering
\includegraphics[width=\textwidth]{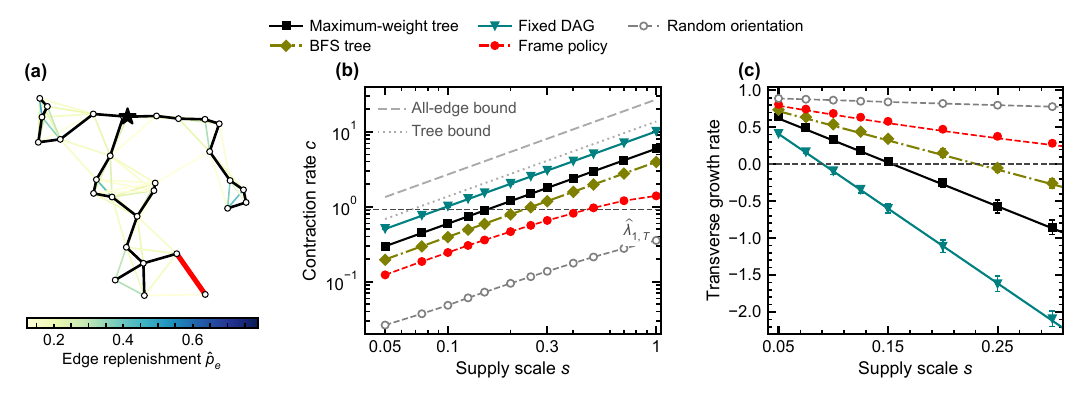}
\caption{Heterogeneous resources on one physical graph. (a) Edge shade represents \(\widehat p_e\); bold edges form the maximum weight tree; the star marks the root, and red marks the bottleneck. (b) Means and standard deviations over eight runs on logarithmic axes. Tree and DAG lines give exact rates; frame and random-orientation lines connect estimates. Gray lines bound acyclic services; the horizontal line estimates isolated growth on \([50,800]\). (c) Means and standard deviations of four variational ODE runs per point. Lines subtract exact tree and DAG rates or estimated frame and random-orientation rates from this empirical growth. Test vectors and resource histories vary between runs. Zero separates local growth from decay, not nonlinear synchronization domains.}
\label{fig:general-graph}
\end{figure*}

For a uniform guarantee on \(\Omega^5\), \eqref{eq:necessary-uniform-resource} requires \(p>a_*/32\) for every feasible directional schedule on this star. Greedy is certified for \(p>\alpha_\Omega/32\); these necessary and sufficient thresholds are approximately \(0.37\) and \(0.38\). For \(\pi_{T,1}\), they are approximately \(0.70\) and \(0.72\). The intervening gaps remain undecided. Policy \(\pi_{T,5}\) cannot meet the necessary requirement even at \(p=1\), although it can synchronize typical trajectories. Figure~\ref{fig:lorenz-numerics}(c) illustrates the domain guarantee for \(\pi_{T,1}\).

Figure~\ref{fig:resource-comparison}(b) integrates the five-node ODE at 20 probabilities near \(0.015\le p\le0.1\), using the same perturbation and 32 histories paired across three policies (1920 runs). At \(p=0.04\), greedy passes in 27 of 32 histories; neither frame policy passes. The transitions require higher \(p\) than the typical-trajectory predictions because the criterion demands tenfold decay within the observation horizon, while intrinsic growth and service counts fluctuate.

Near the transition, increasing the integration resolution from 20 to 80 substeps per event changes finite-horizon contraction fractions by up to 16 percentage points in the tested cells [Fig.~\ref{fig:resource-comparison}(b)].

Figure~\ref{fig:resource-comparison}(c) compares nonlinear and linearized growth on matched finite windows with small initial deviations. Let \(d_i=x_i-x_1\), \(D(t)=\max_{i>1}\|d_i(t)\|_2\), and let \(\Psi\) be the variational flow along the root trajectory. The prediction is \(d_i^{\rm lin}(t)=e^{-\sigma\int_0^t a_i(s)\,ds}\Psi(t,0)d_i(0)\), where \(a_i(s)\in\{0,1\}\) indicates service to leaf \(i\). All 1920 matched windows are resolved. Each lasts at most five time units, ending when either nonlinear or linearized maximum error leaves \([10^{-10},10^{-3}]\).

\subsection{General support with heterogeneous arrivals}
\label{sec:general-graph-numerics}
We next ask how additional incoming links change the resource bottleneck beyond a tree. Figure~\ref{fig:general-graph} considers a geometric graph with 30 nodes in the unit square and 84 edges of length \(d_e<0.3\). Arrivals satisfy \(p_e=0.9s e^{-d_e/0.15}=s\widehat p_e\), where \(s\in[0.05,1]\) scales the supply. We retain the Lorenz parameters from Section~\ref{sec:lorenz-numerics}.

We compare five policies: greedy service on a maximum weight spanning tree, a breadth first search (BFS) tree with the same root, and a fixed DAG; the frame policy \(\pi_{T,1}\) on the maximum weight tree; and greedy service on all edges with directions set by independent random node orderings at each event. The DAG uses a heuristic. For every possible root, we build an ordering by repeatedly adding the node with greatest total connection weight to the placed set, then orient edges forward. We retain the ordering with the largest minimum incoming supply among nonroot nodes.

Theorem~\ref{thm:resource} gives exact greedy rates of approximately \(5.94s\), \(3.96s\), and \(10.11s\) for the two trees and DAG. The DAG gains about \(70\%\) over the maximum weight tree; unbalanced incoming supplies leave it below the bound \(26.94s\). Random orientation gives rate \(0.36\) at \(s=1\), below empirical Lorenz growth.

Resource calculations on 20 further graphs, generated identically conditional on connectivity, give DAG/tree ratios with median \(1.00\), range \(1.00\)--\(2.35\): nine improve, eleven tie. The all-edge bound/DAG ratio has median \(5.11\), range \(2.10\)--\(7.72\).

Panel (b) uses eight 100000-event runs per point. Panel (c) tests the local linear prediction with four transverse variational ODE runs, using event resets and normalization. Isolated growth \(\widehat\lambda_{1,T}\approx0.92\) is estimated on \([50,T]\), \(T=800\). Subtracting exact network rates predicts crossings \(s\approx0.091,0.154,0.231\) for DAG, maximum weight tree, and BFS; interpolated ODE means cross at \(0.091,0.157,0.237\). The frame and random-orientation lines instead subtract estimated long-product rates. Doubling the horizon changes mean growth by up to \(0.06\).

\section{Discussion}
\label{sec:discussion}
\subsection{From temporal requirements to resource design}
Design begins with the contraction needed to overcome intrinsic growth. After checking the supply bound, seek a fixed acyclic orientation with a unique source for which even the least supplied non-source node meets this requirement under greedy service (Figs.~\ref{fig:resource-comparison} and~\ref{fig:general-graph}). The required contraction depends on the target: isolated growth gives a local linear test, whereas the domain growth bound in Theorem~\ref{thm:nonlinear} gives a uniform nonlinear guarantee under linear coupling. Figure~\ref{fig:lorenz-numerics}(c) illustrates the domain guarantee on a star; Fig.~\ref{fig:general-graph}(c) tests the local linear prediction on a larger graph.

When resources are always ready, the inventory constraint disappears, leaving the temporal synchronization problem.

\subsection{Implications for entanglement assisted networks}
A dynamical application requires services that sustain contraction. Stored resources constrain feasible interactions, unlike a prescribed coupling budget \cite{ZhangStrogatz2021}. Comparing contraction with intrinsic growth parallels finite-data-rate control \cite{Nair2007}, but the constrained quantity here is service frequency rather than bit rate. Even with heterogeneous generation, balancing incoming supplies provides an allocation target.

\subsection{Scope}
The resource results assume independent Bernoulli arrivals, lossless storage of one unit per edge, and simultaneous operation of distinct edges. Aging, contention, and correlated arrivals require new arguments. Changing \(\Delta\) for a fixed generation process also changes \(p_e\).

\section{Conclusion}
\label{sec:conclusion}
We have connected the communication needed for synchronization to the resources available to support it. Resource conservation limits contraction, while the weakest incoming supply among non-source nodes determines the optimal greedy rate on a fixed acyclic orientation with a unique source. Balancing these supplies attains the limit. Local contraction criteria and, under linear coupling, guarantees throughout a convex invariant domain translate entanglement service allocation into stability requirements for classical applications. The next challenges are to account for imperfect or delayed samples and to optimize orientations and policies on general graphs.

\section*{Reproducibility}
The scripts, resource histories, derived tables, and validation records supporting this study are available from the corresponding author upon request.

\clearpage
\appendices
\section{Proofs and local regularity conditions}
\label{app:proofs}
\subsection{Proof of Theorem~\ref{thm:temporal-contraction}}
\label{app:exact-proof}

\begin{proof}
Let \(\xi_i(t)\in\R^d\) be the perturbation of node \(i\), and \(\xi_{i,k}=\xi_i(t_k)\). The predictor perturbation is \(\Psi_s(t,t_k)\xi_{j,k}\). In coordinates \(\zeta_i(t)=\Psi_s(t,t_k)^{-1}\xi_i(t)\), the linearization of \eqref{eq:node-ode} becomes
\[
 \dot\zeta_i=-\sigma b\ell_i\zeta_i
       +\sigma b\sum_j[A_k]_{ij}\xi_{j,k}.
\]
Integration gives \(W\), including its identity rows when \(\ell_i=0\). Thus, with \(\Xi_k=\col(\xi_{1,k},\ldots,\xi_{n,k})\),
\begin{equation}
 \Xi_K=(Q_{0:K}\otimes\Psi_s(t_K,0))\Xi_0.
 \label{eq:exact-event-map}
\end{equation}
On the quotient this is \((U^{\mathsf T}Q_{0:K}U)\otimes\Psi_s\). Compression composes because \(Q\one_n=\one_n\). Multiplicativity of the Euclidean Kronecker norm after conjugation by \(I_{n-1}\otimes P^{1/2}\) proves \eqref{eq:exact-norm}. Taking logarithms proves the first inequality and the stated equality condition.

For \(\diamP(X)=\max_{i,j}\|x_i-x_j\|_P\), the diameter contraction inequality and norm equivalence give
\begin{equation}
 \begin{gathered}
 \|Q\|_\perp\le\sqrt{n-1}\,\delta(Q),\\
 \frac{\diamP(X)}{\sqrt2}\le d_P(X,\mathcal D)
 \le\sqrt{\frac{n-1}{2}}\diamP(X).
 \end{gathered}
 \label{eq:distance-diameter-equivalence}
\end{equation}
Here diameter contraction follows by splitting each difference of two stochastic rows into positive and negative parts with equal total mass; their normalized sums are convex combinations of the node vectors. The distance bounds follow from the arithmetic mean and the pairwise variance identity. Together with \eqref{eq:isolated-growth}, these prove \eqref{eq:general-contraction}. If \(G_N\) counts the good blocks in \eqref{eq:block-frequency}, submultiplicativity gives \(\delta(Q_{0:NL})\le\gamma^{G_N}\). A terminal partial block has coefficient at most one, so \eqref{eq:frequency-bound} follows for \(\gamma>0\). If \(\gamma=0\), one good block makes every later prefix coefficient zero, proving the same statement.
\end{proof}

\subsection{Local regularity and proof of Corollary~\ref{cor:nonlinear-synchronization}}
\label{app:nonlinear-proof}

Let \(\mathcal F_{r,\theta}\) be the nonlinear flow from \(t_{rL}\) to \(t_{rL}+\theta\), for \(0\le\theta\le L\Delta\), and set \(\mathcal F_r=\mathcal F_{r,L\Delta}\). Write \(Z_s(t)=\one_n\otimes s(t)\) and \(\mathcal E_P=\{\col(e_1,\ldots,e_n):\sum_i e_i=0\}\). For \(\varepsilon\ge0\), define
\begin{equation}
 \begin{aligned}
 \mathcal T_r(\varepsilon)=\{Z_s(t_{rL})+E:\;&s\in\mathscr S,\ E\in\mathcal E_P,\\
 &\diamP(E)\le\varepsilon\}\cap\mathcal X^n.
 \end{aligned}
 \label{eq:block-tube}
\end{equation}
Assume a common \(r_0>0\), finite \(H_L\), and function \(\omega_L:[0,r_0]\to[0,\infty)\), with \(\omega_L(\varepsilon)\to0\) as \(\varepsilon\downarrow0\), satisfy the following conditions. All block flows from \(\mathcal T_r(r_0)\) exist in \(\mathcal X^n\), and restarting is allowed:
\begin{equation}
 \begin{gathered}
 X\in\mathcal T_r(r_0),\quad\diamP(\mathcal F_r(X))\le r_0\\
 \Longrightarrow\ \mathcal F_r(X)\in\mathcal T_{r+1}(r_0).
 \end{gathered}
 \label{eq:restart-closure}
\end{equation}
Thus the endpoint mean is represented by an isolated trajectory in \(\mathscr S\) at the next boundary. For \(Z=Z_s(t_{rL})\) and every admissible \(Z+E\in\mathcal T_r(r_0)\), require
\begin{equation}
\begin{split}
 \sup_{0\le\theta\le L\Delta}
 \diamP\big(\mathcal F_{r,\theta}(Z+E)-\mathcal F_{r,\theta}(Z)\big)\\
 \le H_L\diamP(E).
\end{split}
\label{eq:within-block-bound}
\end{equation}
\begin{equation}
\begin{split}
 \diamP\big(\mathcal F_r(Z+E)-\mathcal F_r(Z)-D\mathcal F_r(Z)E\big)\\
 \le\omega_L(\diamP(E))\diamP(E).
\end{split}
\label{eq:uniform-remainder}
\end{equation}
Uniformly bounded first and second derivatives on a common compact invariant neighborhood suffice for these two estimates; restart closure remains a separate requirement.

\begin{proof}[Proof of Corollary~\ref{cor:nonlinear-synchronization}]
Let \(D_r=\diamP(X(t_{rL}))\). Exact factorization bounds the block derivative in diameter by \(\beta_L\delta(B_r)\). For sufficiently small common radius, \eqref{eq:uniform-remainder} therefore gives
\[
 D_{r+1}\le[\beta_L\delta(B_r)+\epsilon]D_r.
\]
For \(0\le\gamma<1\), choose \(0<\nu'<\nu\) sufficiently close to \(\nu\), and then \(\epsilon>0\) sufficiently small, so that
\[
 \nu'\log(\beta_L\gamma+\epsilon)
 +(1-\nu')\log(\beta_L+\epsilon)<0.
\]
This is also possible for \(\gamma=0\), since the first logarithm tends to \(-\infty\) as \(\epsilon\downarrow0\). Condition \eqref{eq:block-frequency} gives \(G_N\ge\nu'N\) eventually. The scalar products
\((\beta_L\gamma+\epsilon)^{G_N}(\beta_L+\epsilon)^{N-G_N}\)
then decay exponentially and have a finite maximum \(C\ge1\), including the empty prefix. Choose the initial diameter below the common radius divided by \(C\). Restart closure and induction keep every endpoint within that radius and justify every recurrence. The initial state must also lie in \(\mathcal T_0(r_0)\). The within-block bound extends the decay to continuous time, and \eqref{eq:distance-diameter-equivalence} gives \eqref{eq:nonlinear-target}.
\end{proof}

\subsection{Proof of Theorem~\ref{thm:nonlinear}}
\label{app:domain-proof}

\begin{proof}
Every predictor initialized in \(\Omega\) stays there by \eqref{eq:invariance-assumption}. Convexity of \(\mathcal V\) gives
\(\nabla\mathcal V(x_i)^{\mathsf T}(\widehat y-x_i)\le v_*-\mathcal V(x_i)\).
Consequently, throughout an interval of existence in \(\mathcal X\),
\[
 \frac{d}{dt}(\mathcal V(x_i)-v_*)
 \le-(a+\sigma b\ell_{i,k})(\mathcal V(x_i)-v_*).
\]
Scalar comparison excludes an exit from \(\Omega\). At each event all new predictors start from states in \(\Omega\). Induction, compactness within the open domain, and positive event spacing imply invariance and forward completeness.

The fixed DAG source receives no input, so \(s=x_o\) is isolated. Put \(d_i(t)=\|x_i(t)-s(t)\|_P\), \(d(k)=\col(d_i(t_k))\), and \(D(t)=\max_i d_i(t)\). The one-sided bound implies
\(\|\widehat y_{j\to i,k}(t_k+\theta)-s(t_k+\theta)\|_P
\le e^{\alpha_\Omega\theta}d_j(t_k)\).
This uses the sampled state, even if the sender simultaneously receives input. The upper right Dini derivative satisfies
\[
 D^+d_i\le(\alpha_\Omega-\sigma b\ell_{i,k})d_i
 +\sigma b e^{\alpha_\Omega\theta}\sum_j[A_k]_{ij}d_j(t_k).
\]
At \(d_i=0\) this follows directly from the upper right norm derivative. The sum permits multiple incoming services. Scalar integration, including idle rows, yields
\begin{equation}
 d(t_k+\theta)\le e^{\alpha_\Omega\theta}
 W(A_k;\sigma b\theta)d(k),\qquad0\le\theta\le\Delta.
 \label{eq:continuous-comparison}
\end{equation}
Nonnegativity permits iteration: \(d(K)\le e^{\alpha_\Omega K\Delta}Q_{0:K}d(0)\), with \(d_o=0\).

Order the source first, and let \(\overline Q_{0:K}\) be the principal block of \(Q_{0:K}\) after deleting its source row and column. Write \(\mathbf e_i\) for coordinate vector \(i\) in \(\R^n\). Its norm exponent is \(\Lambda_Q\): if \(F_o\) has rows \(\mathbf e_i^{\mathsf T}-\mathbf e_o^{\mathsf T}\), \(i\ne o\), then \(S_o=F_oU\) is invertible and
\(\overline Q_{0:K}=S_o(U^{\mathsf T}Q_{0:K}U)S_o^{-1}\).
This identity follows from the fixed source row \(\mathbf e_o^{\mathsf T}\). Let \(\|\cdot\|_\infty\) denote the maximum absolute row sum. Fixed norm equivalence and \eqref{eq:nonlinear-rate-condition} imply
\[
 C_0=\sup_{K\ge0}e^{(\alpha_\Omega+\lambda)K\Delta}
 \|\overline Q_{0:K}\|_\infty<\infty.
\]
The iterated comparison bounds \(D(t_K)\) by \(C_0e^{-\lambda t_K}D(0)\). Row stochasticity in \eqref{eq:continuous-comparison} extends this to continuous time with the extra factor \(e^{\max\{0,\alpha_\Omega+\lambda\}\Delta}\).
Finally \(d_P(X,\mathcal D)\le\sqrt{n-1}D\) and \(D\le\sqrt2d_P(X,\mathcal D)\), using respectively the source and the arithmetic mean as reference. The resulting constant is finite and independent of \(X(0)\in\Omega^n\), proving the claim.
\end{proof}

\emph{Proof of the necessary test \eqref{eq:equilibrium-necessity}.}
Perturb \(\one_n\otimes z_*\) by \(\varepsilon v\) in any transverse direction \(v\), admissible for small \(\varepsilon\). For fixed \(K\), the event flow is \(C^1\). Divide the uniform estimate by \(|\varepsilon|\), let \(\varepsilon\to0\), and apply exact factorization:
\[
 \|Q_{0:K}\|_\perp\|e^{K\Delta Df(z_*)}\|_P
 \le C_{\boldsymbol A,\lambda}e^{-\lambda K\Delta}.
\]
The normalized logarithm of the matrix exponential norm tends to \(a_*\), proving the test.\hfill\(\square\)

\subsection{Proof of the resource theorem}
\label{app:resource-proof}
On one probability-one event the strong law holds on every edge:
\(\sum_{k=1}^K g_{e,k}/K\to p_e\).
Fix a realization in this event. Summing \eqref{eq:inventory-update} gives, for every feasible schedule,
\begin{equation}
S_e(K)\le q_{e,0}+\sum_{k=1}^K g_{e,k}-q_{e,K}
\le1+\sum_{k=1}^K g_{e,k}.
\label{eq:resource-conservation}
\end{equation}
This pathwise bound avoids taking probability intersections over policies.

If event \(k\) has \(N_k=\sum_{i,j}[A_k]_{ij}\) services and is acyclic, a topological ordering makes \(W(A_k;\kappa)\) triangular with determinant \(e^{-\kappa N_k}\). The ordering may change with \(k\). In the basis \([\one_n/\sqrt n,U]\), any stochastic product has blocks \(\left(\begin{smallmatrix}1&*\\0&U^{\mathsf T}QU\end{smallmatrix}\right)\). Its quotient determinant is therefore \(\det Q\). The largest singular value is at least their geometric mean, giving
\begin{equation}
\|Q_{0:K}\|_\perp\ge
\exp\left[-\frac{\kappa}{n-1}\sum_{e\in F}S_e(K)\right].
\label{eq:finite-resource}
\end{equation}
Conservation proves the lower limit bound in (i). Combining this finite bound with the exact norm identity also proves \eqref{eq:necessary-resource}, without assuming either exponent has a limit.

For (ii), place \(D\) in a topological ordering with its unique source \(o\) first. Put \(m=n-1\). The root row of every \(W_k\) is its identity row. Subtracting the root coordinate from the other coordinates therefore gives a nonnegative lower triangular matrix \(H_k\), equal to the submatrix of \(W_k\) obtained by deleting the root row and column. Its diagonal is
\[
[H_k]_{ii}=e^{-\kappa\ell_{i,k}},\qquad
\ell_{i,k}=S_i(k+1)-S_i(k).
\]
All its entries lie in \([0,1]\). For \(H_{0:K}=H_{K-1}\cdots H_0\), the fixed root row gives the same similarity \(H_{0:K}=S_o(U^{\mathsf T}Q_{0:K}U)S_o^{-1}\) as in Appendix~\ref{app:domain-proof}. Thus these product norms differ by fixed multiplicative factors.

Suppose the incoming frequencies \(R_i\) exist. Set \(R_*=\min_{i\ne o}R_i\) and
\(E_K=\max_{i\ne o}\max_{0\le t\le K}|S_i(t)-R_it|\), with integer \(t\).
Frequency convergence implies \(E_K=o(K)\): a finite prefix has bounded error, and every later error is at most \(\varepsilon t\) for any fixed \(\varepsilon>0\). The diagonal entries of \(H_{0:K}\) are \(e^{-\kappa S_i(K)}\), giving the lower rate \(-\kappa R_*\) per event. For the upper rate, expand each entry of the product. A nonzero term follows a path of nonroot vertices in \(D\), with at most \(m-1\) strictly increasing index moves and at most \(m\) diagonal intervals. Those intervals total at least \(K-(m-1)\); an interval \([a,b)\) at node \(i\) contributes a count
\[
S_i(b)-S_i(a)\ge R_*(b-a)-2E_K.
\]
Move factors are at most one. The number of vertex paths is finite, and for each path the move times have at most \((K+1)^{m-1}\) choices. Consequently a constant \(C_D\ge1\), independent of \(K\), bounds the number of paths and the conversion to the matrix norm, giving
\begin{equation}
\|H_{0:K}\|_2\le C_D(K+1)^{m-1}
e^{-\kappa R_*K+\kappa R_*(m-1)+2\kappa mE_K}.
\label{eq:dag-upper-proof}
\end{equation}
Taking normalized logarithms gives the matching upper rate, including \(m=1\) and \(R_*=0\). Similarity and \(\kappa/\Delta=\sigma b\) prove \eqref{eq:dag-rate}. On an outward tree, each nonroot node has one incoming edge, so \(R_i=r_{e_i}\), giving the tree specialization.

Under greedy on \(D\), \(s_{e,k}=q_{e,k}\), so \(q_{e,k+1}=g_{e,k+1}\) and
\(S_e(K)=q_{e,0}+\sum_{k=1}^{K-1}g_{e,k}\).
Thus \(S_e(K)/K\to p_e\), giving \eqref{eq:greedy-dag}. For any competitor using the same orientation, choose \(i_*\ne o\) with minimal incoming sum \(\sum_{e\in\operatorname{in}_D(i_*)}p_e\). The product diagonal and fixed similarity imply
\(\|Q_{0:K}\|_\perp\ge C^{-1}e^{-\kappa S_{i_*}(K)}\)
for a fixed \(C>0\). Summing conservation over the incoming edges bounds the competitor's lower limiting exponent from below by \(-\sigma b\sum_{e\in\operatorname{in}_D(i_*)}p_e\), without frequency limits. Greedy attains this value, proving optimality on \(D\).

For the attainment criterion, each allowed edge has exactly one receiver and the source has no incoming edge, so \(\sum_{i\ne o}w_i=\sum_{e\in F}p_e\). The minimum of these \(n-1\) nonnegative numbers equals their average if and only if all are equal. Equations~\eqref{eq:resource-bound} and~\eqref{eq:greedy-dag} then give the stated necessary and sufficient criterion for greedy attainment of the total supply bound.

Finally, any directional service on an undirected tree is acyclic: the tree excludes longer cycles and \eqref{eq:feasible-action} excludes a cycle of length two. With common \(p\), (i) becomes \(\Lambda_Q\ge-\sigma bp\). Greedy attains equality, proving optimality against all directional schedules on that allowed tree.
\hfill\(\square\)

\end{document}


\begin{center}
{\Large\bfseries Numerical Verification}\par
\vspace{5pt}
Supplement to\par
{\itshape Synchronization over Resource-Causal Temporal Networks\par
with Consumable Entanglement Links}
\end{center}

This supplement reports the numerical checks for main Figs.~2--5. Integration and implementation checks are distinguished from the sensitivity of finite-time contraction classifications. All figure numbers below refer to the main manuscript.

\section*{Integration and consistency checks}
Fourth-order Runge--Kutta steps align with service events. Comparisons under step refinement use matched initial conditions and resource histories. The table reports bounds on the largest observed absolute differences in the tested quantities, not error bounds for all trajectories.

\begin{center}
\small
\begin{tabularx}{\linewidth}{@{}Xr@{}}
\toprule
Comparison and measured quantity & Difference below \\
\midrule
Kuramoto, 10 to 20 substeps per event: phase diameter in both service directions & $2\times10^{-9}$ \\
Kuramoto, outward service: phase diameter versus the exact receiver map & $3\times10^{-10}$ \\
Finite-amplitude Lorenz, 20 to 40 substeps: normalized RMS disagreement $E_X(t)/E_X(0)$ & $4\times10^{-7}$ \\
General graph, 20 to 40 substeps: transverse variational growth rate in eight cases on $[50,400]$ & $2\times10^{-5}$ \\
Lorenz star: nonlinear versus linearized growth in 1920 matched small-deviation windows & $4\times10^{-4}$ \\
General graph: direct variational growth versus matched-window matrix factorization & $2\times10^{-5}$ \\
\bottomrule
\end{tabularx}
\end{center}
The finite-amplitude Lorenz comparison concerns normalized disagreement, not individual state trajectories. The matched-window star comparison uses the maximum root-relative Euclidean norm rather than RMS disagreement. The general-graph refinement and factorization comparisons are separate tests of integration sensitivity and the implemented separation, respectively.

\section*{Sensitivity near the Lorenz transition}
For Fig.~4(b), twelve parameter--policy cells near the transition were tested with 32 resource histories each, giving 384 trajectories at each resolution. At observation time 75, increasing from 20 to 80 substeps per event changes 22 of 384 path labels and one of 12 cell-majority labels. The largest contraction-fraction change is about 0.16, or 16 percentage points. The complete pairwise summary is:

\begin{center}
\small
\begin{tabular}{@{}lrrr@{}}
\toprule
Substeps per event & Changed path labels & Changed majority labels & Maximum change \\
\midrule
20 to 40 & 6/384 & 0/12 & 6 percentage points \\
40 to 80 & 26/384 & 1/12 & 19 percentage points \\
20 to 80 & 22/384 & 1/12 & 16 percentage points \\
\bottomrule
\end{tabular}
\end{center}
Percentage-point changes are rounded. The gray segments in Fig.~4(b) span the contraction fractions at all three resolutions. These finite-horizon labels remain resolution sensitive; the small discrepancies in the earlier table do not establish convergence of long-horizon binary classifications. Separately, doubling the observation horizon in the general-graph tests changes mean growth by up to 0.06, as reported with Fig.~5.

\begin{flushleft}
\textbf{Reproducibility.}
The scripts, resource histories, and validation records are available from the corresponding author upon request. The transition statistics are recorded in \nolinkurl{data/lorenz_policy_v4/validation.json} and \nolinkurl{refinement_cell_summary.csv} in the same directory.
\end{flushleft}